\documentclass[]{eptcs}

\providecommand{\event}{WORDS 2011} 

\usepackage{amssymb} 
\usepackage{amsthm}
\usepackage{latexsym} 
\usepackage{amsmath} 
\usepackage{verbatim} 
\usepackage{algorithm} 
\usepackage[noend]{algorithmic} 
\usepackage{graphicx} 
\usepackage{xcolor} 
\usepackage{epsfig}
\usepackage{comment}
\usepackage{cite}

\newtheorem{corollary}{Corollary}

\newtheorem{theorem}{Theorem}
\newtheorem{proposition}{Proposition}

\newtheorem{lemma}{Lemma}

\newcommand{\Sub}{\mbox{\hspace{.01in}Sub}}

\newcommand{\RSub}{\mbox{\hspace{.01in}RSub}}

\renewcommand{\d}{{\ensuremath{\diamond}}}

\title{Recurrent Partial Words\thanks{This material is based upon work supported by the National Science Foundation under Grant No. DMS--0754154. The Department of Defense is also gratefully acknowledged.}}

\author{Francine Blanchet-Sadri
\institute{Department of Computer Science, University of North Carolina,\\
P.O. Box 26170, Greensboro, North Carolina 27402--6170, USA}  
\email{blanchet@uncg.edu}
\and
Aleksandar Chakarov 
\institute{Department of Computer Science, University of Colorado at Boulder,\\
430 UCB, Boulder, Colorado 80309--0430, USA}
\and
Lucas Manuelli 
\institute{Department of Mathematics, Princeton University,\\
Fine Hall, Washington Road, Princeton, New Jersey 08544--1000, USA}
\and
Jarett Schwartz
\institute{Department of Computer Science, Princeton University,\\
35 Olden Street, Princeton, New Jersey 08540--5233, USA}
\and
Slater Stich
\institute{Department of Mathematics, Princeton University,\\
Fine Hall, Washington Road, Princeton, New Jersey 08544--1000, USA}
}

\def\titlerunning{Recurrent Partial Words}
\def\authorrunning{F. Blanchet-Sadri, A. Chakarov, L. Manuelli, J. Schwartz \& S. Stich}

\begin{document}

\maketitle 

\begin{abstract}
Partial words are sequences over a finite alphabet that may contain wildcard symbols, called holes, which match or are compatible with all letters; partial words without holes are said to be full words (or simply words). Given an infinite partial word $w$, the number of distinct full words over the alphabet that are compatible with factors of $w$ of length $n$, called subwords of $w$, refers to a measure of complexity of infinite partial words so-called subword complexity.  This measure is of particular interest because we can construct partial words with subword complexities not achievable by full words. In this paper, we consider the notion of recurrence over infinite partial words, that is, we study whether all of the finite subwords of a given infinite partial word appear infinitely often, and we establish connections between subword complexity and recurrence in this more general framework. 
\end{abstract}

\section{Introduction} 
\label{sec:introduction}

Let $w$ be a (right) infinite word over a finite alphabet $A$. A subword of $w$ is a block of consecutive letters of $w$. The {\em subword complexity} function, $p_w(n)$, counts the number of distinct subwords of length $n$ in $w$. Subword complexity is a well-studied topic and relates to dynamical systems, ergodic theory, theoretical computer science, etc. \cite{All,DBLP:books/daglib/0025558,Cas97,DBLP:journals/dm/Ferenczi99,Ghe}. Another topic of interest on infinite words is the one of {\em recurrence}. An infinite word is said to be recurrent if every subword appears infinitely many times. In 1938, Morse and Hedlund introduced many concepts dealing with recurrence \cite{MoHe38}. Rauzy in \cite{Rau} surveys subword complexity and recurrence in infinite words, while Cassaigne in \cite{DBLP:conf/stacs/Cassaigne01} surveys some results and problems related to recurrence.

Partial words are sequences over a finite alphabet that may contain wildcard symbols, called holes, which match, or are compatible with, all letters in the alphabet (full words are those partial words without holes). Combinatorics on partial words is a relatively new subject \cite{DBLP:journals/tcs/BerstelB99,BSbook}; oftentimes the basic tools have not yet been developed. In \cite{DBLP:conf/tamc/Blanchet-SadriSSW10}, Blanchet-Sadri et al. investigated finite partial words of maximal subword complexity where the subword complexity function of a partial word $w$ over a finite alphabet $A$ assigns to each positive integer, $n$, the number, $p_w(n)$, of distinct full words over $A$ that are compatible with factors of length $n$ of $w$. In \cite{DBLP:conf/lata/ManeaT10}, Manea and Tiseanu showed that computing subword complexity in the context of partial words is a ``hard'' problem.

In \cite{BSChMaScSt_1}, with the help of our so-called hole functions, we constructed infinite partial words $w$ such that $p_w(n)=\Theta(n^\alpha)$ for any real number $\alpha > 1$. In addition, these partial words have the property that there exist infinitely many non-negative integers $m$ satisfying $p_w(m+1)-p_w(m) \geq m^{\alpha}$. Combining these results with earlier ones on full words, we showed that this represents a class of subword complexity functions not achievable by full words. We also constructed infinite partial words with intermediate subword complexity, that is between polynomial and exponential. 

In this paper, we introduce recurrent infinite partial words and show that they have several nice properties. Some of the properties that we present deal with connections between recurrence and subword complexity. Besides reviewing some basics in Section~\ref{sec:preliminaries} and concluding with some remarks in Section~\ref{sec:conclusion}, our paper can roughly be divided into two parts: Among other things, Section~\ref{sec:recurrent_partial_words} extends well-known results on recurrent infinite full words to infinite partial words. Section~\ref{sec:completions} uses the results obtained previously to prove new results. There, we study the relationship between the subword complexity of an infinite partial word $w$ and that of its various completions; here a completion is a ``filling in'' of the holes of $w$ with letters from the alphabet. In particular we ask when can a completion achieve maximal, or nearly maximal, complexity? It turns out that this is intimately related to the notion of recurrence.

\section{Preliminaries} 
\label{sec:preliminaries}
For more information on basics of partial words, we refer the reader to \cite{BSbook}. Unless explicitly stated, $A$ is a finite alphabet that contains at least two distinct letters. We denote the set of all words over $A$ by $A^*$, which under the concatenation operation forms a free monoid whose identity is the empty word $\varepsilon$. 

A \emph{finite partial word} of length $n$ over $A$ is a function $w:\{0,\ldots,n-1\}\rightarrow A\cup \{\d\}$, where $\d \not \in A$. The union set $A\cup \{\d\}$ is denoted by $A_{\d}$ and the length of $w$ by $|w|$. A \emph{right infinite partial word} or \emph{infinite partial word} over $A$ is a function $w:\mathbb{N}\rightarrow A_{\d}$. 
In both the finite and infinite cases, the symbol at position $i$ in $w$ is denoted by $w(i)$. If $w(i) \in A$, then $i$ is defined in $w$, and if $w(i) = \d$, then $i$ is a hole in $w$. If $w$ has no holes, then $w$ is a \emph{full word}. 
A \emph{completion} $\hat{w}$ is a ``filling in'' of the holes of $w$ with letters from $A$. Two partial words $u$ and $v$ are compatible, denoted $u\uparrow v$, if there exist completions $\hat{u}$ and $\hat{v}$ such that $\hat{u} = \hat{v}$.

A finite partial word $w$ over $A$ is said to be $p$-\emph{periodic}, if $p$ is a positive integer such that $w(i) = w(j)$ whenever $i$ and $j$ are defined in $w$ and satisfy $i \equiv j \bmod {p}$. We say that $w$ is \emph{periodic} if it is $p$-periodic for some $p$. An infinite partial word $w$ over $A$ is called \emph{periodic} if there exists a positive integer $p$ (called a \emph{period} of $w$) and letters $a_0, a_1, \ldots, a_{p-1} \in A$ such that for all $i \in \mathbb{N}$ and $j \in \{0,\ldots, p-1\}$, $i \equiv j \bmod{p}$ implies $w(i) \uparrow a_j$. If $w$ is an infinite partial word, then we define the {\em shift} $\sigma_p(w)$ by $\sigma_p(w)(i) = w(i+p)$. The infinite partial word $w$ is called \emph{ultimately periodic} if there exist a finite partial word $u$ and an infinite periodic partial word $v$ (both over $A$) such that $w = uv$. If $w$ is a full ultimately periodic word, then $w=xy^{\omega}=xyyy\cdots$ for some finite words $x, y$ with $y \not = \varepsilon$ called a {\em period} of $w$ (we also call the length $|y|$ a period). If $|x|$ and $|y|$ are as small as possible, then $y$ is called the {\em minimal period} of $w$.  

Given a partial word $w$ over $A$, a finite partial word $u$ is a \emph{factor} of $w$ if there exists some $i \in \mathbb{N}$ such that $u = w(i)\cdots w(i+|u|-1)$. We adopt the following notations for factors: $w(i..j)$ (resp., $w[i..j)$, $w(i..j]$, $w[i..j]$) denotes $w(i+1) \cdots w(j-1)$ (resp., $w(i)\cdots w(j-1)$, $w(i+1)\cdots w(j)$, $w(i)\cdots w(j)$). On the other hand, a finite full word $u$ is a \emph{subword} of $w$, denoted $u\lhd w$, if there exists some $i \in \mathbb{N}$ such that $u\uparrow w[i..i+|u|)$. In the context of this paper, subwords are always finite and full.
We denote by $\Sub_w(n)$ the set of all subwords of $w$ of length $n$, and by $\Sub(w) = \bigcup_{n\geq0}\Sub_w(n)$ the set of all subwords of $w$. Note that $p_w(n)$ is precisely the cardinality of $\Sub_w(n)$.  Furthermore, if $\hat{w}$ is a completion of $w$, then $p_{\hat{w}}(n) \leq p_w(n)$, since $\Sub_{\hat{w}}(n) \subset \Sub_{w}(n)$. 

The following result extends well-known necessary conditions for a function to be the subword complexity function of an infinite full word \cite{DBLP:journals/dm/Ferenczi99}.

\begin{theorem}\label{thm:Fer} The following are necessary conditions for a function $p_w$ from $\mathbb{N}$ to $\mathbb{N}$ to be the subword complexity function of an infinite partial word $w$ over a finite alphabet $A$:
\begin{enumerate}
\item $p_w$ is non-decreasing;
\item $p_w(m+n) \leq p_w(m) p_w(n)$ for all $m,n$;
\item whenever $p_w(n) \leq n$ or $p_w(n+1) = p_w(n)$ for some $n$, then $p_w$ is bounded;
\item if $A$ has $k$ letters, then $p_w(n) \leq k^n$ for all n; if $p_w(n_0) < k^{n_0}$ for some $n_0$, then there exists a real number $\kappa < k$ such that $p_w(n) \leq \kappa^n$ for all $n$ sufficiently large.
\end{enumerate}
\end{theorem}

\section{Recurrent Partial Words}
\label{sec:recurrent_partial_words}

Recurrence is a well-studied topic in combinatorics on infinite full words. We turn our attention to the study of infinite recurrent partial words. We call an infinite partial word $w$ \emph{recurrent} if every $u \in \Sub_w(n)$ occurs infinitely often in $w$; that is, there are infinitely many $j$'s such that $w(j+i) \uparrow u(i)$ for $i \in \{0, \ldots, n-1\}$. We call an infinite partial word $w$ \emph{uniformly recurrent}, if for every $u \in \Sub_w(n)$, there exists $m \in \mathbb{N}$ such that every factor of length $m$ of $w$ has $u$ as a subword, that is, $u \lhd w[0..m-1]$, $u \lhd w[1..m]$, \ldots. Clearly, a uniformly recurrent partial word is recurrent. The following proposition gives a few equivalent formulations of recurrence.

\begin{proposition}\label{pro:recurrence}
	Let $w$ be an infinite partial word. The following are equivalent:
	\begin{enumerate}
		\item The partial word $w$ is recurrent;
		\item Every subword compatible with a finite prefix of $w$ occurs at least twice;
		\item Every subword of $w$ occurs at least twice.
	\end{enumerate}
\end{proposition}
\begin{proof}
	It is clear that $(1)$ implies both $(2)$ and $(3)$, whereas $(3)$ implies $(2)$ since any subword compatible with a finite prefix of $w$ is itself a subword of $w$. To show that $(2)$ implies $(1)$, for the sake of contradiction suppose some word $v\in \Sub_w(n)$ appeared only finitely many times in $w$. Suppose the last occurrence of $v$ starts at position $i$. Then for all $j>i$, $v$ is not compatible with $w[j..j+n)$. Now let $\hat{u}$ be a completion of the prefix of length $i+n$ of $w$ such that $\hat{u}[i..i+n) = v$. Then by $(2)$,  $\hat{u}$ must appear at least twice in $w$. In particular, there exists some position $j>0$ such that $\hat{u} \uparrow w[j..j+i+n)$. But then $v \uparrow w[j+i..j+i+n)$, contradicting the fact that the last occurrence of $v$ started at position $i$. Hence, every subword of $w$ must appear infinitely many times. 
\end{proof}

\begin{theorem}
	\label{thm:rec_up} If $w$ is an infinite recurrent partial word with a positive but finite number of holes, then $w$ is not ultimately periodic. 
\end{theorem}
\begin{proof}
	For the sake of contradiction, suppose $w$ is ultimately periodic. Then we can write $w = xyyy\cdots$ where $y$ is a finite full word such that $|y|$ is the minimal period of $w$. Let $j$ be the position of the last hole in $x$. Let $z = a x[j+1..|x|)y^n = a v y^n$ where $n \geq |y|$ and the letter $a$ is chosen so that $a \neq y(j')$, where $j' = |y|-1-|v|\bmod{|y|}$.
	Since $w$ is recurrent and $z$ is a subword of $w$,  $z$ occurs infinitely many times in $w$. In particular, it occurs somewhere in $u = y^\omega$, where $|y|$ is the minimal period of $u$. Thus, there exists $i \in \{0,\ldots,|y|-1\}$ such that 
	$u(i)\cdots u(i+|z|-1) = z$.
		Since $y(i)=a \neq y(j')$, we have $i \neq j'$.
		
		Set $i' = (i + |v| + 1)\bmod{|y|}$, $y_1=y(0)\cdots y(i'-1)$, and $y_2=y(i')\cdots y(|y|-1)$. We get $y=y_1y_2=y_2y_1$, and so $y_1$ and $y_2$ are powers of a common word $y'$. Thus $y^{|y'|} = (y')^{|y|}$.
	However, $1 \leq |y'| < |y|$. Then $u = y^\omega = (y^{|y'|})^\omega = ((y')^{|y|})^\omega = (y')^\omega$ is $|y'|$-periodic, which contradicts the minimality of period $|y|$. 
\end{proof}

To extend the above theorem to the case where $w$ has infinitely many holes we must introduce some additional restrictions. We would like to impose some constraints on the number of holes and their distribution inside $w$. The motivation for these is the fact that any infinite partial word with a large number of holes exhibits a behavior similar to the one of the trivial partial word $w=\d^\omega$, which is recurrent and periodic.

Next we define the gap function which quantifies the spacing between consecutive appearances of the hole symbol in a partial word. Let $H(n)-1$ be the position of the $n$th hole in an infinite partial word $w$ (we also say that $H(n)$ is the \emph{hole function} of $w$). Then let $h(n)=H(n)-H(n-1)$, for $n\geq2$, be defined as the \emph{gap function} of $w$. For example, the infinite partial word 
\[\d \d a\d a \d aaa \d aaaaa \d aaaaaaaaaaa \d  aaaaaaaaaaaaaaaaaaaa \d \cdots \]
has holes at positions $H(n)-1 = \lceil 2^{4(n-1)/5}\rceil -1$ and the distance between the 5th and 6th holes is $h(6)=H(6)-H(5)=16-10=6$. This is actually an example of an infinite partial word (regarded as a partial word over the alphabet $\{a, b\}$) having a complexity function not achievable by any full word.

\begin{corollary}\label{cor:rec_up}
	Let $w$ be a recurrent partial word with infinitely many holes for which there exists $N>0$ such that $h(n)<h(n+1)$ for all $n\geq N$. Then $w$ is not ultimately periodic.
\end{corollary}
\begin{proof}
	For the sake of contradiction, suppose $w$ is ultimately periodic. Then we can write $w=xy_1y_2\cdots$, where for all $i,j>0$, $y_i$ and $y_j$ are compatible factors of length $p$ with $p$ being the minimal period. We will refer to $y_1, y_2, \ldots$ as the $y$ factors. By choosing sufficiently large $n\geq3$, we can ensure that $h(n)>3p$. Thus, there exists $j>p$ such that both $y_j$ and $y_{j+1}$ are full words. 
	Let $v=xy_1y_2\cdots y_{j-1}$. Then $v$ contains at least two holes. Without loss of generality, assume that $v(l)=v(l')=\d$, for some $l<l'$.

	Let $i_l = (p-|v|+l)\bmod{p}$ and $i_{l'} = (p-|v|+l')\bmod{p}$. Then choose a completion $\hat{v}$ of $v$ such that $\hat{v}(l)\neq y_j(i_l)$ and $\hat{v}(l')\neq y_j(i_{l'})$. Let $u=\hat{v}y_jy_{j+1}$ and $m$ be sufficiently large so that $h(m)>2|u|$. Since $w$ is recurrent, the subword $u$ must occur at some position to the right of $H(m)-1$. So suppose it occurs at position $i$. Then if we let $z=w[i..i+|u|)$ then $z$ contains at most one hole. By the choice of $i_l$ and $i_{l'}$,  at least one of $\hat{v}(l)$ or $\hat{v}(l')$ is incompatible with the corresponding symbol in $z$. Thus the $y$ factors in $u$ cannot align with the $y$ factors in $z$. Also, at least one of the $y$ factors in $z$ is full.
 Analogous to the proof of Theorem~\ref{thm:rec_up}, we conclude that $y_j\cdots y_{j+p-1}$ is periodic with period $p'<p$, where $p'$ is the length of the offset. This contradicts the minimality of $p$ and, therefore, no ultimately periodic words with the desired property exist. 
\end{proof}

Let $w$ be an infinite partial word. We define $R_w(n)$, the \emph{recurrence function} of $w$, to be the smallest integer $m$ such that every factor of length $m$ of $w$ contains at least one occurrence of every subword of length $n$ of $w$. The following theorem extends a well-known result on full words to partial words (see \cite{DBLP:books/daglib/0025558}).

\begin{theorem}
\label{thm:theorem4}
	 Let $w$ be a uniformly recurrent infinite partial word. Then the following hold: 
	\begin{enumerate}
		\item $R_w(n+1)>R_w(n)$ for all $n\geq 0$; 
		\item If for each $n > 0$ there exists an index $i$ such that $w[i..i+n)$ is a full word then $R_w(n) \geq p_w(n)+n-1$ for all $n\geq 0$; 
		\item If $w$ has a positive finite number of holes or an eventually increasing gap function, then $R_w(n)\geq 2n$ for all $n \geq 0$.
	\end{enumerate}
\end{theorem}
\begin{proof}
	 The proof of  $(1)$ is identical to that for full words. For  $(2)$, let $n\geq 0$ and set $m=R_w(n)$. Then there exists an index $i$ such that $v= w[i..i+m)$ is a full word. Since $|v|=m$, $v$ contains every subword of $w$ of length $n$. But any full word of length $m$ contains at most $m-n+1$ distinct subwords of length $n$. Hence, $p_w(n)\leq m-n+1$. Therefore, $R_w(n) \geq p_w(n)+n-1$ for all $n\geq 0$.
For $(3)$, note that the conditions on $w$ together with Theorem \ref{thm:rec_up} and Corollary \ref{cor:rec_up} imply that $w$ is not ultimately periodic. Thus by Theorem~\ref{thm:Fer}(3), $p_w(n)\geq n+1$ for all $n\geq 0$. Since $(3)$ implies $(2)$, we get $R_w(n) \geq p_w(n)+n-1 \geq 2n$ for all $n\geq 0$. 
\end{proof}

The following theorem captures the fact that a uniformly recurrent word cannot achieve maximal complexity.
\begin{theorem}
\label{cannotachieve}
	 Let $w$ be a uniformly recurrent infinite word. Then there exists $N$ such that $p_w(n)<k^n$ for all $n\geq N$, where $k$ is the alphabet size. 
\end{theorem}
\begin{proof}
	 By Theorem~\ref{thm:Fer}(4), we only need to show that $p_w(n)<k^n$ for some $n$. We split the proof into two cases. If $p_w(1)<k$ then we are done. Thus suppose $p_w(1)=k$. Then let $t=R_w(1)$. For the sake of contradiction, suppose $w$ achieves maximal complexity, that is, $p_w(n)=k^n$ for all $n\geq 0$. Then $w$ contains the subword $a^t$, where $a\in A$. Hence, $|a^t|=t=R_w(1)$ implies $b\lhd a^t$ for some $b \in A, b \neq a$, which is a contradiction.
\end{proof}

It is natural to extend the above theorem to partial words with finitely many holes.

\begin{corollary}
	\label{cor:recurrence_complexity} Let $w$ be a uniformly recurrent infinite partial word with finitely many holes. Then there exists $N$ such that $p_w(n)<k^n$ for all $n\geq N$, where $k$ is the alphabet size.
\end{corollary}
\begin{proof}
	 Choose $N$ such that $j\geq N$ implies $w(j)\neq \d$. Then let $v=\sigma_N(w)$. Then uniform recurrence implies that $\Sub(w)=\Sub(v)$. Hence, $p_w(n)=p_v(n)$ and thus Theorem~\ref{cannotachieve} gives us the result. 
\end{proof}

To extend the result to partial words with infinitely many holes we must introduce some additional restrictions. In essence too many holes still allows us to achieve maximal complexity. A trivial example is $w=\d^{\omega}$.
 
\begin{corollary}
	\label{cor:recurrence_complexity2} Let $w$ be a uniformly recurrent infinite partial word for which there exists $n_0$ such that $n\geq n_0$ implies $h(n)\leq h(n+1)$ and $\lim_{n \to \infty}h(n+1)-h(n) = \infty$. Then there exists $N>0$ such that $p_w(n)<k^n$ for all $n \geq N$, where $k$ is the alphabet size. 
\end{corollary}
\begin{proof}
	 The proof is very similar to that of Theorem~\ref{cannotachieve}. 
	
\end{proof}

The following result illustrates the relationship between a recurrent partial word and its completions.

\begin{proposition}
\label{pro:prop2}
	Let $w$ be an infinite partial word having a finite number of holes or
				an eventually increasing gap function.
	Then $w$ is recurrent if and only if every completion $\hat{w}$ is recurrent.
\end{proposition}
\begin{proof}
	First suppose $w$ is recurrent. Let $\hat{w}$ be any completion of $w$. Proposition~\ref{pro:recurrence} implies that we only need to show that each subword of $\hat{w}$ appears at least twice. Choose $u \in \Sub_{\hat{w}}(n)$. Suppose $w$ has a finite number of holes. Then there exists $N>0$ such that $j\geq N$ implies $w(j)\neq \d$. Since $w$ is recurrent, $u$ appears starting at some position $i\geq N$, that is, $u = w[i..i+n)$. But note that $\hat{w}[i..i+n) = w[i..i+n)$. Hence, the subword $u$ occurs twice in $\hat{w}(i)\hat{w}(i+1)\cdots$ and thus $\hat{w}$ is recurrent.
	
	Now suppose $w$ has an eventually increasing gap function. Since $w$ is recurrent, we see that there exists a word $v$ such that $uvu\in \Sub(w)$. Let $m=|uvu|$ and choose $N$ such that for all $j\geq N$ we have $h(j)>m$. Recurrence implies that $uvu$ appears starting at some position $i$ greater than $H(N)$. Suppose $uvu\uparrow z$ where $z=w[i..i+m)$. Then $z$ contains at most one hole. Hence, at least one of $u = w[i..i+n)$ or $u = w[i+m-n..i+m)$ holds. Since $w$ is recurrent, $uvu$ has to appear again in $w$, and so $u$ must appear one more time in a factor of $w$ that contains no holes. Without loss of generality, assume that $w[i'..i'+n)$ is the desired full word. Then $w[i'..i'+n) = \hat{w}[i'..i'+n)$. Hence $u = \hat{w}[i'..i'+n)$ so that $u$ appears at least twice in $\hat{w}$. Hence, $\hat{w}$ is recurrent.
		
	Now suppose every completion is recurrent. Choose $u\in \Sub(w)$. Then there exists a completion $\hat{w}$ such that $u \in \Sub(\hat{w})$. Since $\hat{w}$ is recurrent $u$ occurs again at an index different from where it appeared initially in $w$. Suppose $u = \hat{w}[i..i+|u|)$. Since $\hat{w}$ is a completion of $w$ we see that $\hat{w}[i..i+|u|) \uparrow w[i..i+|u|)$. Hence $u$ occurs twice in $w$ so that $w$ is recurrent. 
\end{proof}

\section{Completions of Infinite Partial Words}
\label{sec:completions}

We investigate the relationship between the complexity of an infinite partial word $w$ and the complexity achievable by a given completion $\hat{w}$. Our main question is given an infinite partial word $w$ how much complexity can be preserved while passing to a completion?

\begin{theorem} \label{partialEquiv}
Let $w$ be an infinite  
recurrent partial word. Then there exists a completion of $w$, $\hat{w}$, such that $\Sub(w) = \Sub(\hat{w})$.
\end{theorem}
\begin{proof}
The set $\Sub(w)$ is countable, so choose some enumeration of its elements $x_0, x_1, x_2, \ldots$.  
Choose $n_0$ so that $x_0 \lhd w[0..n_0]$. Since $x_1$ occurs infinitely often in $w$, we can find some $n_1>n_0$ so that $x_1 \lhd w(n_0..n_1]$. Similarly we can find some $n_2 > n_1$ so that $x_2 \lhd w(n_1..n_2]$ and so on for each $x_i$. Now we complete $w[0..n_0]$ so that it contains $x_0$ as a subword, $w(n_0..n_1]$ so that it contains $x_1$, and so on to get $\hat{w}$. By construction $\Sub(w) \subset \Sub(\hat{w})$ and we have  $\Sub(\hat{w}) \subset \Sub(w)$. 
\end{proof}

Another question is to ask when a completion with maximal complexity exists. We know by Theorem~\ref{partialEquiv} that it is sufficient that the original partial word $w$ be recurrent. In the case where $w$ has infinitely many holes, this turns out to be necessary as well.

\begin{theorem}\label{thm:completion}
	Let $w$ be a partial word with infinitely many holes. Then $w$ is recurrent if and only if there exists a completion $\hat{w}$ such that $\Sub(w) = \Sub(\hat{w})$.
\end{theorem}
\begin{proof}
	The forward implication is simply a consequence of Theorem \ref{partialEquiv}. 
For the backward implication, suppose there exists a completion $\hat{w}$ such that $\Sub(w) = \Sub(\hat{w})$. We show that the prefix of length $H(n)-1$ of $\hat{w}$ occurs twice for every $n\geq 1$. Choose $a\in A$ such that $a \neq \hat{w}(H(n)-1)$. Then $v= \hat{w}[0..H(n)-1)a \in \Sub(w) = \Sub(\hat{w})$. Hence $v$ must occur somewhere in $\hat{w}$. But it cannot occur as a prefix since $a \neq \hat{w}(H(n)-1)$. Thus there exists $i> 0$ such that $\hat{w}[i..i+H(n)) = v$. But then $\hat{w}[i..i+H(n)-1) = \hat{w}[0..H(n)-1)$ so that $\hat{w}[0..H(n)-1)$ appears twice. Thus every prefix of $\hat{w}$ occurs twice and thus $\hat{w}$ is recurrent and since $\Sub(w) = \Sub(\hat{w})$, $w$ is recurrent as well. 
\end{proof}

The proof really relies on the fact that $w$ has infinitely many holes. The theorem is not true in the case of finitely many holes. For example, choose $w = {\d} a^\omega$ and $\hat{w} = ba^\omega$. Then $\Sub(w) = \Sub(\hat{w})$ but $w$ is not recurrent since $b$ occurs only once. However, we note that $\sigma(w)$ is recurrent. This fact actually holds more generally. First we call an infinite partial word $w$ {\em ultimately recurrent} if there exists an integer $p \geq 0$ such that $\sigma_p(w)$ is recurrent. With this definition in hand we can extend Theorem~\ref{thm:completion} to the case when we may not have infinitely many holes.

\begin{corollary}
\label{cor:unlabelled1}
	Let $w$ be an infinite partial word with at least one hole. If there exists a completion $\hat{w}$ of $w$ such that $\Sub(w) = \Sub(\hat{w})$, then $w$ is ultimately recurrent. In fact $\sigma_{H(1)}(w)$ is recurrent, where $H(n)$ is the hole function.
\end{corollary}

\begin{proof}
	We claim that if $H(1)-1$ is the position of the first hole and $p = H(1)$ then $\sigma_p(w)$ is recurrent. Let $v = \sigma_p(w)$. By Proposition \ref{pro:recurrence}, it suffices to show that every finite prefix of any completion of $v$ occurs twice in $v$. Thus suppose $z$ is a full word such that $z \uparrow v[0..n)$. Choose a completion $u$ of $w[0..n+p)$ so that $z = u[p..n+p)$. In addition, we require that the hole at position $H(1)-1$ be filled in such a way that $u(H(1)-1) \neq \hat{w}(H(1)-1)$. Then $u \in \Sub(\hat{w})$. However, we see that the way we filled in the hole at $H(1)-1$ prohibits $u$ from occuring as a prefix of $\hat{w}$. Thus there exists an index $i > 0$ such that $u = \hat{w}[i..i+n+p)$. But then $z = u[p..n+p) =\hat{w}[i+p..i+n+p) \uparrow v[i..i+n)$ so that $z$ appears twice in $v$. Hence $v$ is recurrent. Thus $w$ is ultimately recurrent.
\end{proof}

	 Let $\RSub_w(n)$ denote the set of recurrent subwords of length $n$ of a finite or infinite partial word $w$. Let $\RSub(w) = \bigcup_{n\geq 1} \RSub_w(n)$. Let $r_w(n) = |\RSub_w(n)|$ and $d_w(n) = p_w(n) - r_w(n)$. In other words, $d_w(n)$ counts the number of non-recurrent subwords of length $n$.
Note that $d_w(n)$ is non-decreasing. The following proposition captures the fact that in an ultimately recurrent partial word with finitely many holes almost every subword is recurrent.
\begin{proposition}\label{pro:bounded}
	Let $w$ be an infinite partial word with finitely many holes. Then $w$ is ultimately recurrent if and only if $d_w(n)$ is bounded.
\end{proposition}
\begin{proof}
	Suppose $w$ is ultimately recurrent. Then there exists $p$ such that $\sigma_p(w)$ is recurrent. We claim that $d_w(n) \leq p$. Note that any subword beginning at an index $\geq p$ must be recurrent. Thus any non-recurrent subword must appear starting at a position less than $p$. Each position $i$ with $0\leq i < p$ contributes finitely many distinct subwords of length $n$.
	
	Now suppose $d_w(n)$ is bounded. Since $d_w(n)$ is non-decreasing, there exist a constant $C$ and an integer $n$ such that $C = d_w(n) = d_w(m)$ for all $m\geq n$. Since there are only $C$ non-recurrent subwords of length $n$ and each appears only finitely many times in $w$, there exists an $N$ such that none of these non-recurrent subwords appear starting at positions $i \geq N$. We claim that $w' = \sigma_N(w)$ is recurrent. For the sake of contradiction, suppose $w'$ is not. Then there must exist a non-recurrent word $v$ in $\Sub(w')$. Assume without loss of generality that $|v| = m \geq n$. Now we break the proof into two cases. If the prefix of length $n$ of $v$ was a non-recurrent subword of $w$, then this would contradict the choice of $N$. So suppose that the prefix of length $n$ of $v$ is not a non-recurrent subword of $w$. Note that each length $n$ non-recurrent subword contributes at least one distinct length $m$ non-recurrent subword. In addition $v$ is distinct from each of these since the prefixes of length $n$ do not match. Thus $d_w(m) > d_w(n)$, a contradiction. 
\end{proof}

The case when $w$ has infinitely many holes is markedly different. In particular $d_w(n)$ cannot be positive and bounded. This is captured in the following proposition.

\begin{proposition}
\label{pro:prop4}
	Let $w$ be a partial word with infinitely many holes. Then $d_w(n)$ is either identically zero or unbounded.
\end{proposition}
\begin{proof}
	For the sake of contradiction, suppose there exists a constant $C$ such that $1\leq d_w(n) \leq C$ for all $n>0$. Then there exists an $n$ and a $v$ such that $v \in \Sub_w(n) \backslash \RSub_w(n)$. Since $v\in \Sub(w)$ there exists an index $i$ such that $v \uparrow w[i..i+n)$. Since $w$ has infinitely many holes, there exists an $m$ such that $w[n..m]$ has at least $h$ holes where $k^h > C$. Since $v$ is not recurrent, each of the completions of $vw[n..m]$ is non-recurrent. Hence if we let $j = |vw[n..m]|$ we see that $d_w(j) = p_w(j) - r_w(j) \geq k^h > C$, a contradiction.
\end{proof}

The infinite partial word $w$ being ultimately recurrent does not imply anything about the growth of $d_w(n)$ by itself. However, we can relate the growth of $p_w(n)$ and $r_w(n)$. Intuitively we can think of $w$ being ultimately recurrent as capturing the fact that $w$ has a large proportion of recurrent subwords. We would expect that $r_w(n)$ is a good approximation of $p_w(n)$. In fact, it turns out that $p_w(n) = \Theta(r_w(n))$.

\begin{proposition}
\label{pro:prop5}
	Let $w$ be an ultimately recurrent infinite partial word. Then there exists a constant $C$ such that $r_w(n) \leq p_w(n) \leq C r_w(n)$ for all $n$ sufficiently large. In other words, $p_w(n) = \Theta(r_w(n))$.
\end{proposition}
\begin{proof}
	Suppose $N$ is such that $\sigma_N(w)$ is recurrent. Consider $d_w(n)$ for $n > N$. Then every non-recurrent subword of length $n$ must start at some position $i$, $0\leq i < N$, and must be compatible with a factor of the form $w[i..i+n)$. We can break the factor into two parts: $w[i..N)$ which may have a non-recurrent completion, and $w[N..i+n)$ where every completion is recurrent. If there are $h$ holes in $w[i..N)$, there are at most $k^h$ completions of $w[i..N)$. Any completion of $w[N..i+n)$ must be recurrent, each has length at most $n$, so there are at most $r_w(n)$ such completions.  Hence there are at most $k^h r_w(n)$ distinct non-receurrent subwords of length $n$ starting at position $i$. Since there are exactly $N$ possible starting positions for non-recurrent subwords, we see that $d_w(n) \leq N k^h r_w(n)$. Since $p_w(n) = r_w(n) + d_w(n)$, the result follows. 
\end{proof}

One might expect that if $w$ has a large proportion of recurrent subwords then it might be ultimately recurrent. However, this is not true in general. Consider the word $w$ that is all $a$'s except for $b$'s at positions $H(n)-1 = n^2-1$. Then it is easy to check that $p_w(n)$ is linear. Also, it is clear that every subword containing at most one $b$ is recurrent. There are $n+1$ such length $n$ words. Hence both $r_w(n)$ and $p_w(n)$ are linear. However, $w$ is not ultimately recurrent since any subword with at least two $b$'s occurs exactly once. Thus the requirement that a word be ultimately recurrent is too restrictive. In fact we can also find a partial word with infinitely many holes such that the same property holds. All that is required is to let the hole function be $H(n) = \lceil \alpha^n \rceil$ with $\alpha >2$ being a real number, and then notice that $p_w(n)$ is asymptotically linear. Then as before every word with at most one $b$ is recurrent so that $r_w(n) = n+1$. Hence $p_w(n) \leq C r_w(n)$ for a suitable constant $C \in \mathbb{R}$.

The above proposition has an easy corollary. We know that we can always find a completion that contains all the recurrent subwords. Thus if $w$ is ultimately recurrent then there exists a completion $\hat{w}$ whose complexity function is of the same order of growth as that of $w$.

\begin{corollary}
	Let $w$ be an ultimately recurrent infinite partial word. Then there exists a completion $\hat{w}$ such that $p_w(n) = \Theta(p_{\hat{w}}(n))$.  
\end{corollary}

Intuitively, the ``closest'' that a complexity function can be to another is to be within a constant of that function. Thus, if we could not attain maximal complexity with a completion, the best we could hope for is ``off by a constant'' complexity. The following proposition shows that this is not possible in general.

\begin{proposition}\label{pro:close}
	Let $w$ be a partial word with infinitely many holes. If $\hat{w}$ is a completion of $w$ such that $p_w(n) \leq p_{\hat{w}}(n) + C$ for all $n>0$ and some constant $C$, then $\Sub(w) = \Sub(\hat{w})$ and thus $p_w(n) = p_{\hat{w}}(n)$.
\end{proposition}
\begin{proof}
	For the sake of contradiction, assume there existed $v \in \Sub_w(n)$ with $v \not\in \Sub_{\hat{w}}(n)$. Then $v \uparrow w[i..i+n)$ for some  $i$. Now using the fact that $w$ contains infinitely many holes, we can chose an index $m$ such that $w_n \cdots w_m$ has at least $h$ holes where $k^h > C$. Then there are at least $k^h$ completions of $vw_n \cdots w_m$. Since $p_w(n) - p_{\hat{w}}(n) \leq C$ at least one of these completions, call it $u$, must also be a subword of $\hat{w}$. But then since $v$ is a prefix of $u$ we would necessarily have $v \in \Sub({\hat{w}})$, a contradiction.
\end{proof}

Thus a completion $\hat{w}$ must satisfy either $p_w(n) = p_{\hat{w}}(n)$ or the function $f(n) = p_w(n) - p_{\hat{w}}(n)$ must be unbounded. The above result actually holds more generally. What allows us to prove the above proposition is that we are able to use the holes to create ``enough'' subwords to overcome the constant $C$. Thus if we have $p_w(n) \leq p_{\hat{w}}(n) + \varphi(n)$ for some increasing function $\varphi$, then as long as the holes are spaced close enough together we must have $p_w(n) = p_{\hat{w}}(n)$. Thus the closer spaced the holes become, the farther away a non-maximal completion must be in terms of complexity. 

\begin{proposition}\label{pro:varphi}
	Let $w$ be an infinite partial word with hole function $H(m)$. If $\hat{w}$ is a completion of $w$ such that $p_w(n) \leq p_{\hat{w}}(n) + \varphi(n)$ for all $n > 0$ and some increasing function $\varphi$ satisfying $\lim_{n \rightarrow \infty} \frac{\varphi(H(n))}{k^n} = 0$, then $p_w(n) = p_{\hat{w}}(n)$.
\end{proposition}
\begin{proof}
	The proof follows the same general strategy as that of Proposition \ref{pro:close}. For the sake of contradiction, suppose there existed $v \in \Sub_w(n)$ such that $v \not\in \Sub_{\hat{w}}(n)$. Then $v \uparrow w[i..i+n)$ for some $i$. Let $j$ be the smallest integer such that $H(j) \geq i+n$. Then choose $m>j$ such that $k^{m-j} > \varphi(H(m))$. Then there are at least $k^{m-j}$ distinct completions of $vw[i+n..H(m))$. Since they have length less than $\varphi(H(m))$ and $\varphi$ is increasing we see that at least one of them, call it $u$, must be contained in $\Sub(\hat{w})$. But since $v$ is a prefix of $u$ we see that $v \in \Sub(\hat{w})$, a contradiction. 
\end{proof}

The situation is different for infinite partial words with finitely many holes. If $w$ has finitely many holes then for each completion there exists a constant $C$ such that $p_w(n) \leq p_{\hat{w}}(n) + C$. However, if $C$ is small enough then it turns out that $w$ is actually ultimately recurrent.

\begin{proposition}
\label{pro:prop8}
	Let $w$ be an infinite partial word with exactly $h$ holes where $1 \leq h < \infty$. If there exists a completion $\hat{w}$ of $w$ such that $p_w(n) \leq p_{\hat{w}}(n) + C$ for all $n>0$ and some constant $C$ satisfying $C \leq k^h - 2$, then $w$ is ultimately recurrent. 
\end{proposition}
\begin{proof}
	We show that $v = \sigma_{H(h)}(w)$ is recurrent. We show that every finite prefix of $v$ occurs at least twice. Consider $v[0..n)$. Then there are $k^h$ distinct completions of $w[0..H(h))v[0..n)$. Since $C \leq k^h - 2$ at least two of these completions must be subwords of $\hat{w}$. Thus at least one is not compatible with a prefix of $\hat{w}$. Let $u$ be this subword. Then there must exist some $i>0$ such that $u = \hat{w}[i..i+|u|)$. Since $v[0..n)$ is a suffix of $u$ this implies that there exists $j>0$ such that $v[0..n) = v[j..j+n)$ so that every finite prefix of $v$ occurs twice.
\end{proof}

The following is a strengthening of Theorem \ref{thm:completion}.
\begin{theorem}\label{thm:recurrence_completion3}
	Let $w$ be a partial word with infinitely many holes. Then $w$ is recurrent if and only if there exists a completion $\hat{w}$ and constant $C$ such that $p_w(n) \leq p_{\hat{w}}(n) + C$ for all $n > 0$.
\end{theorem}
\begin{proof}
	The forward implication is a direct consequence of Theorem \ref{thm:completion}.
	For the backward implication, if $p_w(n) \leq p_{\hat{w}}(n) + C$ then Proposition \ref{pro:close} implies that $\Sub(w) = \Sub(\hat{w})$. Then Theorem \ref{thm:completion} implies that $w$ is recurrent. 
\end{proof}

Intuitively Theorem~\ref{thm:recurrence_completion3} shows us that we cannot get too close (i.e off by a constant) to the complexity of $w$ with a completion unless $w$ is recurrent. In fact the conditions in the previous theorems are actually stronger than what is needed. In order to show recurrence of $w$ we only need to be able to find completions $\hat{w}$ that stay close to $p_w(n)$ for arbitrarily large $n$. This is made precise in the following lemma.

\begin{lemma}\label{lem:recurrence_almost_maximal}
	Let $w$ be a partial word with infinitely many holes. Suppose that for each $N > 0$ there exists a completion $\hat{w}$ such that $p_w(n) = p_{\hat{w}}(n)$ for all $n \leq N$. Then $w$ is recurrent.
\end{lemma}
\begin{proof}
	By Proposition \ref{pro:recurrence} it suffices to show that each subword appears at least twice. We argue by contradiction. Suppose there exists a subword $v \in \Sub_w(n)$ that appears only once. Say $v \uparrow w[i..i+n)$. Since $w$ has infinitely many holes there exists a smallest index $j \geq i+n$ such that $w(j) = \d$. Now choose a completion $\hat{w}$ such that $p_w(m) = p_{\hat{w}}(m)$ for all $m \leq j+1$. Choose $a \in A$ such that $a \neq \hat{w}_j$. Then consider $u = v w[i+n..j) a$. Then $|u| \leq j+1$ so that $u$ is a subword of $\hat{w}$. Thus $u$ must appear somewhere in $\hat{w}$. But it cannot appear starting at position $i$. Thus there must exist another position $i'$ such that $u = \hat{w}[i'..i'+|u|)$. But then $v \uparrow w[i'..i'+n)$ so that $v$ appears twice in $w$, a contradiction.
\end{proof}

We can now use the above lemma to prove a stronger version of Theorem \ref{thm:recurrence_completion3}. 

\begin{corollary}
\label{cor:unlabelled2}
	Let $w$ be a partial word with infinitely many holes. Suppose there exists a constant $C$ such that for each $N > 0$ there exists a completion $\hat{w}$ such that $p_w(n) \leq p_{\hat{w}}(n) + C$ for all $n \leq N$. Then $w$ is recurrent.
\end{corollary}
\begin{proof}
	We reduce the proof to an application of Lemma \ref{lem:recurrence_almost_maximal}. For each $n >0$, we  find a completion $\hat{w}$ such that $p_w(n) = p_{\hat{w}}(n)$ which allows us to apply the lemma. Fix $n$. Now choose $N$ such that all subwords of $w$ of length $n$ appear in $w[0..N)$. Now choose $M$ such that $w[N..M)$ has at least $h$ holes where $k^h > C$. Then choose $\hat{w}$ such that $p_w(m) \leq p_{\hat{w}}(m) + C$ for all $m \leq M$. Now we claim that $p_w(n) = p_{\hat{w}}(n)$. Choose $v\in \Sub_w(n)$. Now complete $w[0..N)$ such that $v$ appears as a subword. Call this completed subword $u$. Then there are at least $k^h > C$ completions of $uw[N..M)$. Hence since $p_w(n) \leq p_{\hat{w}}(n) + C$ at least one completion must be a subword of $\hat{w}$. Since $v$ is a prefix of $u$ this implies $v \in \Sub(\hat{w})$. Thus $\Sub_w(n) = \Sub_{\hat{w}}(n)$ and hence $p_w(n) = p_{\hat{w}}(n)$. All that remains is to apply the lemma to conclude that $w$ is recurrent. 
\end{proof}

A similar argument provides a generalization of Proposition \ref{pro:varphi}.

\begin{proposition}\label{pro}
	Let $w$ be an infinite partial word with hole function $H(m)$ and let $\varphi$ be an increasing function. If for each $N>0$ there exists a completion $\hat{w}$ such that $p_w(n) \leq p_{\hat{w}}(n) + \varphi(n)$ for all $n \leq N$ and $\lim_{n \rightarrow \infty} \frac{\varphi(H(n))}{k^n} = 0$, then $p_w(n) = p_{\hat{w}}(n)$ and $w$ is recurrent.
\end{proposition}

Another question that one may ask is how the complexity of a completion $p_{\hat{w}}(n)$ relates to the recurrence function $r_w(n)$ for the original partial word $w$. If the complexity of all completions is bounded by $r_w(n)$ (up to a constant) then it turns out that $w$ is actually ultimately recurrent. The following theorem states this rigorously.

\begin{theorem}\label{thm:ultimate_recurrence}
	Let $w$ be an infinite partial word. Then $w$ is ultimately recurrent if and only if for each completion $\hat{w}$ there exists a constant $C$ such that $p_{\hat{w}}(n) \leq r_w(n) + C$ for all $n > 0$.
\end{theorem}
\begin{proof}
	Suppose $w$ is ultimately recurrent. Then there exists $C$ such that $\sigma_C(w)$ is recurrent. Then consider any completion $\hat{w}$. Any subword starting at an index $i \geq C$ is contained in $\RSub(w)$. Thus the only possible subwords in $\Sub(\hat{w}) \backslash \RSub(w)$ must occur starting at positions $0 \leq i < C$. There are at most $C$ such subwords. Thus $p_{\hat{w}}(n) \leq r_w(n) + C$. Now suppose for each completion $\hat{w}$ there exists a constant $C$ such that $p_{\hat{w}}(n) \leq r_w(n) + C$ for all $n > 0$. The intuition of the proof is as follows. If $w$ is not ultimately recurrent we can find as many non-recurrent subwords as we like. This allows us to find a completion $\hat{w}$ that contains all the recurrent subwords and have $p_w(n) - r_w(n)$ be unbounded. 
	
	For the sake of contradiction, assume $w$ is not ultimately recurrent. Then let $\{ w_n \}$ be an enumeration of the elements of $\RSub(w)$. Since $w$ is not ultimately recurrent we can choose a non-recurrent subword $v_0$. Let $i_0$ be an index such that $v_0$ is not a subword of $\sigma_{i_0}(w)$. Complete $w[0..i_0)$ so that it contains $v_0$ as a subword. Now choose $j_0$ such that $w_0$ is a subword of $w[i_0..j_0)$. Complete $w[i_0..j_0)$ so that $w_0$ is a subword. Now since $w$ is not ultimately recurrent there exists a non-recurrent subword $v_1$ appearing in $\sigma_{j_0}(w)$ with $|v_1| \geq |v_0|$. Choose $i_1$ such that $v_1$ is not a subword of $\sigma_{i_1}(w)$. Complete $w[j_0..i_1)$ such that $v_1$ appears as a subword. Now choose $j_1$ such that $w_1$ is a subword of $w[i_1..j_1)$ and complete it so that $w_1$ appears as a subword. Continuing on in this way we see that $\hat{w}$ contains all the recurrent subwords and infinitely many non-recurrent subwords. Now fix a $C$. Choose $m = |v_C|$. Then each $v_i$ for $0\leq i \leq C$ contributes (by extending to the right) a length $m$ subword. In addition each of these is non-recurrent. Also they are all distinct since otherwise they would have to have matching prefixes, a contradiction. Hence $p_{\hat{w}}(m) \geq C + 1 + r_w(m)$. Thus for this completion there exists no constant $C$ such that $p_{\hat{w}}(n) \leq r_w(n) + C$ for all $n>0$, a contradiction.
\end{proof}	

We can actually strengthen the above theorem. The proof above shows that if $w$ is ultimately recurrent then the same $C$ works for all completions $\hat{w}$. In other words the bound is uniform across completions. We state this in a corollary.

\begin{corollary}
	Let $w$ be an infinite partial word. If $w$ is ultimately recurrent, then there exists a constant $C$ such that $p_{\hat{w}}(n) \leq r_w(n) + C$ for all $n>0$ and all completions $\hat{w}$ of $w$.
\end{corollary}

Oftentimes if every completion of an infinite partial word $w$ has a certain property, then $w$ has it as well. In particular this property holds with respect to ultimate recurrence.

\begin{proposition}
\label{pro:prop10}
	Let $w$ be an infinite partial word. Then $w$ is ultimately recurrent if every completion $\hat{w}$ is ultimately recurrent. 
\end{proposition}
\begin{proof}
	If $w$ is not ultimately recurrent, then the completion constructed in the proof of Theorem \ref{thm:ultimate_recurrence} is not ultimately recurrent. 
\end{proof}

We now introduce the notion of a {\em most complex completion}. The motivation is that this concept helps us understand the role of recurrent subwords in completions. Let $w$ be an infinite partial word. We say that $\hat{w}$ is a most complex completion of $w$ if for all completions $\bar{w}$ of $w$ and all $n>0$ we have $p_{\bar{w}}(n) \leq p_{\hat{w}}(n)$. In general a most complex completion of an infinite partial word may not exist. However, assuming that $w$ possesses such a completion we have the following result which states that a most complex completion must contain all the recurrent subwords. The intuition here is straightforward. In a rough sense one gets the recurrent subwords of $w$ for free. We can delay putting them in the completion for arbitrarily long, and they still occur after that for us to capture. Thus it is not difficult to construct a completion of higher complexity if this is not the case.

\begin{proposition}
\label{pro:prop11}
	Let $w$ be an infinite partial word. If $\hat{w}$ is a most complex completion, then  $\RSub(w) \subset \Sub(\hat{w})$.
\end{proposition}


\section{Conclusion}
\label{sec:conclusion}

Intuitively all the above work culminates to show that completions can achieve complexities equal (or ``close'') to that of the original partial word if and only if the word is recurrent or ultimately recurrent. Another interesting avenue of research would be to investigate whether a relation exists between the growth of $r_w(n)$ and that of $p_w(n)$. Although it would be nice, the answer seems to be no. Given any constant $\delta < 1$ we can find a partial word with infinitely many holes such that $\frac{r_w(n)}{p_w(n)} \rightarrow \delta$. Also, even if we impose the restriction that $r_w(n)$ be linear then we still have a fair bit of freedom with the complexity of $p_w(n)$. In particular we can make it so that asymptotically $p_w(n)$ attains any polynomial complexity. We can also attain some intermediate complexities, i.e. functions of the form $2^{\sqrt{n}}$. The construction of these examples is actually quite simple. You just have a word that is all $a$'s with holes at positions $H(n)-1$. Since the hole functions in all of our constructions are eventually increasing we see that any word with at least two $b$'s is not recurrent. Since there are exactly $n+1$ words of length $n$ with at most one $b$ we see that $r_w(n) = n+1$. By controlling the growth of $H(n)$ we can control the growth of $p_w(n)$. The slower $H(n)$ grows the faster $p_w(n)$ grows.

\bibliographystyle{eptcs}
\bibliography{bibliography}

\begin{thebibliography}{10}
\providecommand{\bibitemdeclare}[2]{}
\providecommand{\urlprefix}{Available at }
\providecommand{\url}[1]{\texttt{#1}}
\providecommand{\href}[2]{\texttt{#2}}
\providecommand{\urlalt}[2]{\href{#1}{#2}}
\providecommand{\doi}[1]{doi:\urlalt{http://dx.doi.org/#1}{#1}}
\providecommand{\bibinfo}[2]{#2}

\bibitemdeclare{article}{All}
\bibitem{All}
\bibinfo{author}{Jean-Paul Allouche} (\bibinfo{year}{1994}):
  \emph{\bibinfo{title}{Sur la complexit\'{e} des suites infinies}}.
\newblock {\sl \bibinfo{journal}{Bulletin of the Belgium Mathematical Society}}
  \bibinfo{volume}{1}, pp. \bibinfo{pages}{133--143}.

\bibitemdeclare{book}{DBLP:books/daglib/0025558}
\bibitem{DBLP:books/daglib/0025558}
\bibinfo{author}{Jean-Paul Allouche} \& \bibinfo{author}{Jeffrey~O. Shallit}
  (\bibinfo{year}{2003}): \emph{\bibinfo{title}{Automatic Sequences - Theory,
  Applications, Generalizations}}.
\newblock \bibinfo{publisher}{Cambridge University Press}.
\newblock
  \urlprefix\url{http://www.cambridge.org/gb/knowledge/isbn/item1170556/?site_%
locale=en_GB}.

\bibitemdeclare{article}{DBLP:journals/tcs/BerstelB99}
\bibitem{DBLP:journals/tcs/BerstelB99}
\bibinfo{author}{Jean Berstel} \& \bibinfo{author}{Luc Boasson}
  (\bibinfo{year}{1999}): \emph{\bibinfo{title}{Partial Words and a Theorem of
  Fine and Wilf}}.
\newblock {\sl \bibinfo{journal}{Theor. Comput. Sci.}}
  \bibinfo{volume}{218}(\bibinfo{number}{1}), pp. \bibinfo{pages}{135--141}.
\newblock \urlprefix\url{http://dx.doi.org/10.1016/S0304-3975(98)00255-2}.

\bibitemdeclare{book}{BSbook}
\bibitem{BSbook}
\bibinfo{author}{Francine Blanchet-Sadri} (\bibinfo{year}{2008}):
  \emph{\bibinfo{title}{Algorithmic Combinatorics on Partial Words}}.
\newblock \bibinfo{publisher}{Chapman \& Hall/CRC Press},
  \bibinfo{address}{Boca Raton, FL}.

\bibitemdeclare{misc}{BSChMaScSt_1}
\bibitem{BSChMaScSt_1}
\bibinfo{author}{Francine Blanchet-Sadri}, \bibinfo{author}{Aleksandar
  Chakarov}, \bibinfo{author}{Lucas Manuelli}, \bibinfo{author}{Jarett
  Schwartz} \& \bibinfo{author}{Slater Stich} (\bibinfo{year}{2010}):
  \emph{\bibinfo{title}{Constructing partial words with subword complexities
  not achievable by full words}}.
\newblock \bibinfo{note}{Preprint}.

\bibitemdeclare{inproceedings}{DBLP:conf/tamc/Blanchet-SadriSSW10}
\bibitem{DBLP:conf/tamc/Blanchet-SadriSSW10}
\bibinfo{author}{Francine Blanchet-Sadri}, \bibinfo{author}{Jarett Schwartz},
  \bibinfo{author}{Slater Stich} \& \bibinfo{author}{Benjamin~J. Wyatt}
  (\bibinfo{year}{2010}): \emph{\bibinfo{title}{Binary De Bruijn Partial Words
  with One Hole}}.
\newblock In: {\sl \bibinfo{booktitle}{TAMC}}, pp. \bibinfo{pages}{128--138}.
\newblock \urlprefix\url{http://dx.doi.org/10.1007/978-3-642-13562-0_13}.

\bibitemdeclare{article}{Cas97}
\bibitem{Cas97}
\bibinfo{author}{Julien Cassaigne} (\bibinfo{year}{1997}):
  \emph{\bibinfo{title}{Complexit\'{e} et facteurs sp\'{e}ciaux}}.
\newblock {\sl \bibinfo{journal}{Bulletin of the Belgium Mathematical Society}}
  \bibinfo{volume}{4}(\bibinfo{number}{1}), pp. \bibinfo{pages}{67--88}.

\bibitemdeclare{inproceedings}{DBLP:conf/stacs/Cassaigne01}
\bibitem{DBLP:conf/stacs/Cassaigne01}
\bibinfo{author}{Julien Cassaigne} (\bibinfo{year}{2001}):
  \emph{\bibinfo{title}{Recurrence in Infinite Words}}.
\newblock In: {\sl \bibinfo{booktitle}{STACS}}, pp. \bibinfo{pages}{1--11}.
\newblock \urlprefix\url{http://dx.doi.org/10.1007/3-540-44693-1_1}.

\bibitemdeclare{article}{DBLP:journals/dm/Ferenczi99}
\bibitem{DBLP:journals/dm/Ferenczi99}
\bibinfo{author}{S{\'e}bastien Ferenczi} (\bibinfo{year}{1999}):
  \emph{\bibinfo{title}{Complexity of sequences and dynamical systems}}.
\newblock {\sl \bibinfo{journal}{Discrete Mathematics}}
  \bibinfo{volume}{206}(\bibinfo{number}{1-3}), pp. \bibinfo{pages}{145--154}.
\newblock \urlprefix\url{http://dx.doi.org/10.1016/S0012-365X(98)00400-2}.

\bibitemdeclare{article}{Ghe}
\bibitem{Ghe}
\bibinfo{author}{Irina Gheorghiciuc} (\bibinfo{year}{2007}):
  \emph{\bibinfo{title}{The subword complexity of a class of infinite binary
  words}}.
\newblock {\sl \bibinfo{journal}{Advances in Applied Mathematics}}
  \bibinfo{volume}{39}, pp. \bibinfo{pages}{237--259}.

\bibitemdeclare{inproceedings}{DBLP:conf/lata/ManeaT10}
\bibitem{DBLP:conf/lata/ManeaT10}
\bibinfo{author}{Florin Manea} \& \bibinfo{author}{Catalin Tiseanu}
  (\bibinfo{year}{2010}): \emph{\bibinfo{title}{Hard Counting Problems for
  Partial Words}}.
\newblock In: {\sl \bibinfo{booktitle}{LATA}}, pp. \bibinfo{pages}{426--438}.
\newblock \urlprefix\url{http://dx.doi.org/10.1007/978-3-642-13089-2_36}.

\bibitemdeclare{article}{MoHe38}
\bibitem{MoHe38}
\bibinfo{author}{Marston Morse} \& \bibinfo{author}{Gustav~A. Hedlund}
  (\bibinfo{year}{1938}): \emph{\bibinfo{title}{Symbolic dynamics}}.
\newblock {\sl \bibinfo{journal}{American Journal of Mathematics}}
  \bibinfo{volume}{60}, pp. \bibinfo{pages}{815--866}.

\bibitemdeclare{article}{Rau}
\bibitem{Rau}
\bibinfo{author}{G.~Rauzy} (\bibinfo{year}{1982--83}):
  \emph{\bibinfo{title}{Suites \`{a} termes dans un alphabet fini}}.
\newblock {\sl \bibinfo{journal}{S\'{e}minaire de Th\'{e}orie des Nombres de
  Bordeaux}} \bibinfo{volume}{25}, pp. \bibinfo{pages}{2501--2516}.

\end{thebibliography}
\end{document}